\documentclass[10pt,a4paper]{article}
\usepackage[latin1]{inputenc}
\usepackage{amsmath}
\usepackage{amsfonts}
\usepackage{amssymb}
\usepackage{graphicx}

\usepackage{cite}
\usepackage{amsmath,amssymb,amsfonts}
\usepackage{algorithmic}
\usepackage{graphicx}
\usepackage{textcomp}
\usepackage{xcolor}
\usepackage{hyperref}
\usepackage[ruled,linesnumbered]{algorithm2e}
\usepackage{amsthm}
\usepackage{verbatim}	
\usepackage{fancyhdr}	
\newtheorem{definition}{Definition}
\newtheorem{theorem}{Theorem}

\providecommand{\keywords}[1]
{
	\small	
	\textbf{\textit{Keywords---}} #1
}

\date{}

\title{FSPVDsse: A Forward Secure Publicly Verifiable Dynamic SSE scheme}

\author{Laltu Sardar$^{1}$  and Sushmita Ruj$^{1,2}$\\
	$^{1}$Indin Statistical Institute, Kolkata, India\\
	$^{2}$CSIRO Data61, Australia\\
	E-mail: laltuisical@gmail.com, sushmita.ruj@csiro.au
}

\begin{document}
	\maketitle
	
	\begin{abstract}
		A symmetric searchable encryption (SSE) scheme allows a client (data owner) to search on encrypted data outsourced to an untrusted cloud server. 
		The search may either be a single keyword search or a complex query search like conjunctive or Boolean keyword search. Information leakage is quite high for dynamic SSE, where data might be updated.  It has been proven that to avoid this information leakage an SSE scheme with dynamic data must be \emph{forward private}. A  dynamic SSE scheme is said to be {forward private}, if adding a keyword-document pair does not reveal any information about the previous search result with that keyword.
		
		In SSE setting, the data owner has very low computation and storage power. In this setting, though some schemes achieve forward privacy with honest-but-curious cloud, it becomes difficult to achieve forward privacy when the server is malicious, meaning that it can alter the data. 
		Verifiable dynamic SSE requires the server to give a proof of the result of the search query. The data owner can verify this proof efficiently. 
		In this paper, we have proposed a generic publicly verifiable dynamic SSE (DSSE) scheme that makes any forward private DSSE scheme verifiable without losing forward privacy. The proposed scheme does not require any extra storage at owner-side and requires minimal computational cost as well for the owner. Moreover, we have compared our scheme with the existing results and show that our scheme is practical. 
		
	\end{abstract}
	
	 \keywords{Searchable encryption, Forward privacy,  Verifiability,  BLS signature, Cloud computing.}
	
	\section{Introduction}
	Data stored in untrusted servers is prone to attacks by the server itself. In order to protect confiential infomation, clients store encrypted data. This makes searching on data quite challenging. 
	%
	A searchable symmetric encryption (SSE) scheme enables a client or data owner to store its data in a cloud server without loosing the ability to search over them. 
	When an SSE scheme supports update, it is called a dynamic SSE (DSSE) scheme.
	
	There are plenty of works on SSE 
	as well as DSSE. 
	Most of them considers the cloud server to be honest-but-curious. An honest-but-curious server follows the protocol but wants to  extract information about the plaintext data and the queries. 
	However, if the cloud itself is malicious, it does not follow the protocol correctly. In the context of search, it can return only a subset of results, instead of all the records of the search. 
	So, there is need to verify the results returned by the cloud to the querier. 
	An SSE scheme for static data where the query results are verifiable is called Verifiable SSE (VSSE). 
	Similarly, if the data is dynamic the scheme is said to be a verifiable dynamic SSE (VDSSE).   
	
	There are single keyword search VSSE schemes which are either new constructions supporting verifiability or design techniques to achieve verifiability on the existing SSE schemes by proposing generic algorithm.
	VSSE with single keyword search has been studied in \cite{icc/ChaiG12}, \cite{ ccs/ChengYGZR15}, \cite{fgcs/LiuLLJL18}. In \cite{infocom/SunLLH015}, \cite{esorics/Wang0SLAZ18} etc., VSSE scheme with conjunctive query has been studied. 
	Moreover, there are also works that gives VDSSE scheme for both single keyword search (\cite{isci/MiaoWWM19}) as well as complex query search including fuzzy keyword search (\cite{trustcom/ZhuLW16}) and Boolean query (\cite{globecom/JiangZGL15}).
	However, Most of them are \emph{privately verifiable}. 
	A VSSE or VDSSE scheme is said to be privately verifiable if only querier, who receive search result, can verify it. On the other hand, a VSSE or VDSSE scheme is said to be \emph{publicly verifiable} if any third party, including the database owner, can verify the search result without knowing the content of it. 
	
	There is also literature on public verifiability. 
	Soleimanian and Khazaei~\cite{dcc/SoleimanianK19} and Zhang et al.~\cite{icws/0016XYL16} have presented SSE schemes which are publicly verifiable. 
	VSSE with Boolean range queries  has been studied by Xu et al.~\cite{corr/abs-1812-02386}. 
	Though, their verification method is public,  since the verification is based over blockchain databases, it has extra monetary cost. Besides, Monir Azraoui \cite{cns/AzraouiEOM15} presented a conjunctive search scheme that is publicly verifiable.
	In case of dynamic database, publicly verifiable scheme by  Jiang et al.~\cite{globecom/JiangZGL15} supports Boolean Query and that by Miao et al.~\cite{isci/MiaoWWM19} supports single keyword search.
	
	However, file-injection attack \cite{usenix/ZhangKP16}, in which the client encrypts and stores files sent by the server, recovers keywords from future queries, has forced researchers to think about dynamic SSE schemes to be forward private where adding a keyword-document pair does not reveal any information about the previous search result with that keyword. 
	In addition, in presence of malicious cloud server, the owner can outsource the verifiabilty to a third party auditor to reduce its computational overhead.
	The only forward private single keyword search VSSE scheme is proposed by Yoneyama and Kimura~\cite{icics/YoneyamaK17}. However, the scheme is privately verifiable and the owner requires significant amount of computation for verification.
	
	\subsection{Our Contribution}
	In this paper, we have contributed the followings in the literature of VSSE.    
	\begin{enumerate}
		\item We have formally define a verifiable DSSE scheme.  Then we have proposed a generic verifiable SSE scheme  ({\color{blue} $\Psi_s$})  which is very efficient and easy to integrate.
		\item We have proposed a generic publicly verifiable dynamic SSE scheme ({\color{blue} $\Psi_f$}). Our proposed scheme is forward private. This property is necessary to protect a DSSE scheme from file injection attack. However, no previous publicly verifiable scheme is forward private. In fact, only forward private scheme  \cite{icics/YoneyamaK17} is privately verifiable.
		\item We present formal security proofs for these schemes and shows that they are adaptively secure in random oracle model. 		   
	\end{enumerate}
	Both of the schemes do not uses any extra storage, at owner side, than the embedded schemes. Thus, for a resource constrained client, the schemes are very effective and efficient.
	
	In Table~\ref{tab_diff_schemes}, we have compared our proposed schemes with existing ones.
	\begin{table}[!htbp]
		\centering
		\caption{Different verifiable SSE schemes} \label{tab_diff_schemes} \medskip
		\resizebox{\linewidth}{!}{
			\begin{tabular}{|c|c|c|c|c|c|c|c|c|}
				\hline
				Data Type 	& 			\multicolumn{4}{|c|}{static} 				& 				\multicolumn{4}{|c|}{dynamic}				\\ \hline
				Query Type	& \multicolumn{2}{|c|}{single}&\multicolumn{2}{|c|}{complex}& \multicolumn{2}{|c|}{single}&\multicolumn{2}{|c|}{complex}\\ \hline
				Verification& private 		& public 	& private 		& public 		& private 	& public 		& private		& public \\ \hline
				Schemes		&  \cite{icc/ChaiG12},  \cite{ccs/ChengYGZR15}, \cite{fc/OgataK17}, \cite{fgcs/LiuLLJL18}, {\color{blue} $\Psi_s$}
				& \cite{dcc/SoleimanianK19}
				& \cite{esorics/Wang0SLAZ18}, \cite{ijisec/LiZQLX18}, \cite{corr/abs-1812-02386}
				& \cite{dcc/SoleimanianK19}
				& \cite{icics/YoneyamaK17}, \cite{iacr/BostFP16}
				& \cite{isci/MiaoWWM19}, { \color{blue} $\Psi_f$}
				& \cite{trustcom/ZhuLW16}
				& \cite{globecom/JiangZGL15} \\ \hline	
				Forward Private & \multicolumn{4}{|c|}{not applicable} &\multicolumn{4}{|c|}{ \cite{icics/YoneyamaK17}, \color{blue} $\Psi_f$}\\ \hline 	
			\end{tabular}
		}
	\end{table}
	
	\subsection{Organization}
	We have briefly described the works related to verifiable SSE in Section~\ref{sec_RelatedWorks}. 
	We have discussed the required preliminary topics in Section~\ref{sec_Preliminaries}.
	In Section~\ref{sec_VSSE}, we have presented a generic approach of verifiable SSE scheme.
	In Section~\ref{sec_proposedScheme}, we present our proposed generic construction of publicly verifiable DSSE scheme in details. 
	We have compared its complexity with similar publicly verifiable schemes in Section~\ref{sec_comparison}. 
	Finally, we summaries our work in Section~\ref{Conclusion} with possible future direction of research.

	\section{Related Works} \label{sec_RelatedWorks}
	The term \emph{Searchable Symmetric Encryption} is first introduced by Curtmola et al.~\cite{ccs/CurtmolaGKO06} where they have given formal definition of keyword search schemes over encrypted data.
	Later, Chase et al.~\cite{asiacrypt10} and  Liesdonk et al.~\cite{sdm10}  presented  single keyword search SSE for static database. 
	Thereafter, as the importance of database updating is increased, the work has been started on dynamic SSE.
	Kamara et al. \cite{ccs/KamaraPR12} first have introduced a dynamic single keyword search scheme based on encrypted inverted index. 
	There are remarkable works on single keyword search on dynamic database. 
	However, file-injection attack, by Zhang et al.~\cite{usenix/ZhangKP16} have forced the researchers to think about dynamic SSE schemes to be forward private. 
	It is easy to achieve forward privacy with ORAM. 
	However, due to large cost of communication, computation and storage, ORAM based schemes are almost impractical.
	
	In 2016 Bost~\cite{ccs/Bost16} has presented a non-ORAM based forward private dynamic SSE scheme. 
	Later, few more forward private schemes have been proposed. 
	Though, the works \cite{ccs/BostMO17}, \cite{ccs/SunYLSSVN18} etc. provide backward privacy, now we are not bother about it since there is no formal attack on non-backward private DSSE schemes.
	Though, till now there are no formal attack on non-backward private DSSE schemes, there are works \cite{ccs/BostMO17} and \cite{ccs/SunYLSSVN18} that provide backward privacy.
	In most of the above mentioned schemes, the cloud service providers are considered to be honest-but-curious.  
	However, the schemes fails to provide security in presence of malicious cloud server.
	
	Chai and Gong~\cite{icc/ChaiG12} have introduced the first VSSE scheme. They stores the set of document identifiers in a trie like data structure where each node corresponding to some keyword stores identifiers containing it.
	Cheng et al.~\cite{ccs/ChengYGZR15} have presented a VSSE scheme for static data based on  the secure indistinguishability obfuscation. Their scheme also supports 
	Boolean queries and provides publicly verifiability on the return result. 
	Ogata and Kurosawa~\cite{ccs/ChengYGZR15} have presented a no-dictionary generic verifiable SSE scheme. Cuckoo hash table is used here for this private verifiable scheme. With  multi-owner setting,  Liu  et al.~\cite{fgcs/LiuLLJL18} have presented a VSSE with aggregate keys.  Miao et al.~\cite{chinaf/MiaoMLZL17} presented a VSSE in same multi-owner setting.
	However, all of the above schemes were for static database and are privately verifiable where
	the VSSE schemes by Soleimanian and Khazaei~\cite{dcc/SoleimanianK19} and Zhang et al.~\cite{icws/0016XYL16} are publicly verifiable. 
	
	The above works are only for static data. There are few works also that deals with complex queries when the data is static. 
	Conjunctive query on static data has been studied by Sun et al.~\cite{infocom/SunLLH015}, Miao et al.~\cite{ppna/MiaoMWLWL17}, Wang et al.~\cite{esorics/Wang0SLAZ18}, Li et al.~\cite{ijisec/LiZQLX18},  Miao et al.~\cite{percom/MiaoMLJZSL17} etc. These schemes have private verifiability. Boolean range queries on SSE has been studied by Xu et al.~\cite{corr/abs-1812-02386}. Though, their verification method is public,  since the verification is based over blockchain databases it has good monetary cost. Besides, Monir Azraoui \cite{cns/AzraouiEOM15} presented a conjunctive search that is publicly verifiable.
	
	Dynamic verifiable SSE with complex queries also has been studied. Zhu et al.~\cite{trustcom/ZhuLW16} presented a dynamic fuzzy keyword search scheme which is  privately verifiable and 
	Jiang et al.~\cite{globecom/JiangZGL15} has studied Publicly Verifiable Boolean Query on dynamic database. 
	Moreover, single keyword search scheme  on dynamic data  is described by  Yoneyama and Kimura~\cite{icics/YoneyamaK17},
	Bost et al.~\cite{iacr/BostFP16} etc.
	
	A publicly verifiable SSE scheme is recently also proposed by Miao et al.~\cite{isci/MiaoWWM19}.  
	Yoneyama and Kimura~\cite{icics/YoneyamaK17} presented a scheme based on Algebraic PRF which is verifiable as well as forward private that performs single keyword search. However, the scheme is privately verifiable and the owner requires significant amount of computation for verification.
	
	Our proposed scheme $\Psi_f$ is generic forward private verifiable scheme which is compatible with any existing forward private DSSE scheme. 
	Our scheme also do not use any extra owner-storage for verifiability and has minimal search time computation for the owner.   
	
	\section{Preliminaries}	\label{sec_Preliminaries}
	
	\subsection{Cryptographic Tools}
	
	\subsubsection{Bilinear Map} \label{ss:BilinearMaps}
	Let $\mathbb{G}$ and $\mathbb{G}_T$ be two (multiplicative) cyclic groups of prime order $q$. Let $\mathbb{G}=<g>$. A map $\hat{e} :\mathbb{G} \times \mathbb{G} \rightarrow \mathbb{G}_T$ is said to be an \emph{admissible non-degenerate bilinear map} if--
	a) $\hat{e}(u^a,v^b) = \hat{e}(u,v)^{ab}$, $\forall u,v \in \mathbb{G}$ and $\forall a, b \in \mathbb{Z}$ (bilinearity)
	b) $ \hat{e}(g,g) \neq 1$ (non-degeneracy)
	c) $\hat{e}$ can be computed efficiently.
	
	\subsubsection{Bilinear Hash}
	Given a bilinear map $\hat{e}: \mathbb{G}\times \mathbb{G} \rightarrow \mathbb{G}_{T}$ and a generator $g$, a bilinear hash  $\mathcal{H} : \{0,1\}^* \rightarrow \mathbb{G} $ maps every random string to an element of $\mathbb{G}$. The map is defined as $\mathcal{H}(m) = g^m,\ \forall m\in \{0,1\}^* $.
	
	\subsubsection{Bilinear Signature (BLS)}
	Let $\hat{e}: \mathbb{G} \times \mathbb{G} \rightarrow \mathbb{G}_{T}$ be a  bilinear map where $|\mathbb{G}|=|\mathbb{G}_{T}|=q$, a prime and $ \mathbb{G}=<g>$. A bilinear signature (BLS) scheme $\mathcal{S}$=$(\mathtt{Gen}$, $\mathtt{Sign}$, $ \mathtt{Verify})$ is a tuple of three algorithms as follows.\vspace{-10pt}
	\begin{itemize}
		\item $ (sk,pk)\gets \mathtt{Gen}$: It selects $ \alpha \xleftarrow{\$}[0,q-1] $. It keeps the private key $sk=\alpha$. publishes the public key $ pk=g^{\alpha}$.
		\item $\sigma  \gets \mathtt{Sign} (sk,m)$: Given $ sk=\alpha$, and some message $m$, it outputs the signature $\sigma=(\mathcal{H}(m))^\alpha = {(g^m)}^\alpha$ where $\mathcal{H} : \{0,1\}^*\rightarrow\mathbb{G}$ is a bilinear hash.
		\item $\{0/1\} \gets \mathtt{Verify}(pk,m,\sigma)$: Return whether $ \hat{e}(\sigma ,g)= \hat{e}(\mathcal{H}(m),g^{\alpha})$
	\end{itemize}

	\subsection{System Model}
	In this section, we briefly describe the system model considered in this paper. In our model of verifiable SSE, there are three entities--Owner, Auditor and Cloud.
	The system model is shown in the Fig.~\ref{fig_systemModel}.  We briefly describe them as follows.  
	\begin{figure}[!htbp]
		\centering
		{\includegraphics[width=0.5\textwidth]{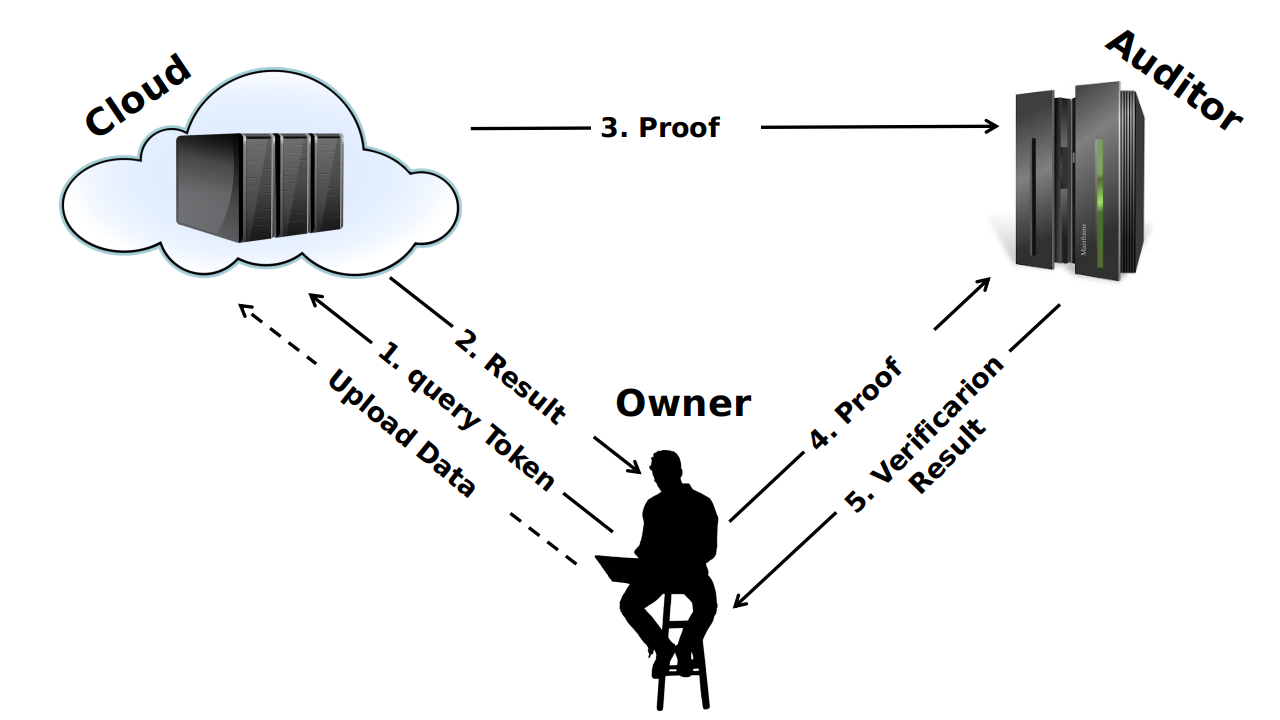}}
		\caption{The system model}
		\label{fig_systemModel}
	\end{figure}
	\begin{enumerate}
		\item {\bf Owner:} Owner is the owner as well as user of the database. It is considered to be \emph{trusted}. It builds an secure index, encrypts the data and then outsources both to the cloud. Later, it sends encrypted query to the cloud for searching. Therefore, it is the querier as well. It is the client who requires the service.
		\item {\bf Cloud:} Cloud or the cloud server is the storage and computation service provider. It stores the encrypted data sent from the owner and gives result of the query requested by it. The cloud is assumed to be \emph{malicious}. It can deviate from protocol by not only computing on, or not storing the data but also making the querier fool by returning incorrect result.
		\item {\bf Auditor:} Auditor is an \emph{honest-but-curious} authority which does not collude with the cloud. Its main role is to verify whether the cloud executes the protocol honestly. It tells the querier whether the returned result is correct or not.
	\end{enumerate}

	\subsection{Design Goals} 
	Assuming the above system model, we aim to provide solution of the verifiability problem of existing forward private schemes.  In our design, we take care to achieve the following objectives.
	\begin{enumerate}
		\item {\bf Confidentiality:} The cloud servers should not get any information about the uploaded data. On the other hand, queries should not leak any information about the database. Otherwise the cloud may get knowledge about the plaintext information. 
		
		\item  {\bf  Efficiency:} In our model, the cloud has a large amount of computational power as well as good storage. The owner is weak. So, in the scheme the owner should require significantly small amount of computation and storage cost while performing verifiability. 
		
		\item 	{\bf Scalability:} Since, the owner have to pay for the service provided by the cloud, it is desirable to outsource as much data as possible. The owner should capable to outsource large amount of data to the cloud. On the other hand, the cloud should answer the queries fast using less computation power.
		
		\item {\bf Forward privacy:} It is observed previously that a DSSE scheme without forward privacy is vulnerable to even honest-but-curious adversary. So, our target is to make a publicly verifiable DSSE scheme without loosing its forward privacy property.
	\end{enumerate}
	
	\subsection{Definitions}
	Let $\mathcal{W}$ be a set of keywords. $\mathcal{D} $ be the space of document identifiers and $\mathcal{DB}$ be the set of documents to be outsourced. 
	Thus, $\mathcal{DB} \subseteq \mathcal{D} $.  
	For each keyword $w\in \mathcal{W}$, the set of document identifiers that includes ${w}$ is denoted by $DB(w)=\{id^{w}_{1},id^{w}_{1}, \ldots, id^{w}_{c_w} \}$, where $c_w = |DB(w)|$ and $id^{w}_{i} \in \mathcal{DB} $. 
	Thus, $ \bigcup \limits _{w\in \mathcal{W}} DB(w) \subseteq \mathcal{DB} $. 
	Let $\overline{DB} = \{ {c_{id}}: id \in \mathcal{D} \}$ where $c_{id}$ denotes the encrypted document that has identifier $id$.
	
	We assume that there is a one-way function $H'$ that maps each identifier $id$ to certain random numbers. These random numbers is used as document name corresponding to the identifier. The function is can be computed by both the owner and cloud. However, from a document name, the identifier can not be recovered. Throughout, we use identifiers. However, when we say cloud returns documents to the owner, we assume the cloud performs the function on every identifiers before returning them.       
	
	Let, $H: \{0,1\}^* \rightarrow \{0,1\}^\lambda$ be a cryptographic hash function, 
	$ \mathcal{H}$ be a bilinear hash, 
	$R: \{0,1\}^* \rightarrow \{0,1\}^* $ be a PRNG and 
	$F : \{0,1\}^\lambda \times \{0,1\}^* \rightarrow \{0,1\}^\lambda $ be a HMAC.
	A \emph{stateful algorithm} stores its previous states and use them to compute the current state.

	\subsection{Verifiable Dynamic Searchable Symmetric Encryption (VDSSE)}
	An SSE scheme allows a client to outsource a dataset it owns to a cloud service provider in encrypted form without loosing the ability to perform query over the data. The most popular query is the keyword search where the dataset is a collection of documents. The client can retrieve partial encrypted data without revealing any meaningful information to the cloud. Throughout we take query as single keyword search query.
	
	
	A \emph{dynamic SSE} (DSSE) scheme is a SSE scheme that supports updates. 
	A \emph{Verifiable DSSE} (VDSSE) scheme is a DSSE scheme together with verifiability. 
	The verification can be done either by an external auditor or the owner. 
	The primary reason to bring a auditor is to reduce computational costs of verifiability at owner-side. This allows an owner to be lightweight.
	
	Though a VDSSE scheme supports update, we do not verify whether the cloud updates the database correctly or not. We only want to get the correct result with respect to current state of the database. If cloud updates the database incorrectly, it can not give the actual result. Due to verifiability, it will be failed in verification process to the auditor.  
	We define a verifiable DSSE scheme formally as follows.

	\begin{definition}[Verifiable Dynamic SSE] \label{def_vdsse}
		A verifiable dynamic SSE (VDSSE) scheme $\Psi$ is a tuple $(\mathtt{VKeyGen}$, $\mathtt{VBuild}$, $\mathtt{VSearchToken}$, $\mathtt{ VSearch}$, $\mathtt{VUpdateToken}$, $\mathtt{VUpdate})$ of algorithms defined as follows.
		\begin{itemize}	
			\item $K \leftarrow \mathtt{VKeyGen}(1^\lambda)$: It is a probabilistic polynomial-time (PPT)  algorithm run by the owner. Given security parameter $\lambda$ it outputs a key $K$.
			
			\item $(\overline{DB} , \gamma) \leftarrow \mathtt{VBuild}(K,\mathcal{DB})$: The owner run this PPT algorithm. Given a key $K$ and a set of documents $\mathcal{DB}$, it outputs the encrypted set of documents $\overline{DB}$ and an encrypted index $\gamma $.  
			
			\item $\tau_s \leftarrow \mathtt{VSearchToken}(K,w)$: On input a keyword $w$ and  the key $K$, the owner runs this PPT algorithm to output a search token $\tau_s$.
			
			\item $(R_w, \nu_w )\leftarrow  \mathtt{VSearch}(t_s,\gamma)$: It is a PPT algorithm run by the cloud and the auditor collaboratively that returns a set of document identifiers result $R_w$ to the owner with verification bit $\nu _w$. 
			
			\item $\tau_u \leftarrow  \mathtt{VUpdateToken}(K,id)$: It is a owner-side PPT  algorithm that takes the key $K$ and a document identifier $id$  and outputs a update token $\tau_u$. 
			
			\item $(\overline{DB}' ,\gamma ') \leftarrow  \mathtt{VUpdate}(\tau_u, op, \gamma, \overline{DB})$: It is a PPT algorithm run by the cloud. It takes an update token $\tau_u$, operation bit $op$, the encrypted document set $\overline{DB}$ and the index $\gamma$ and outputs updated  $(\overline{DB}', \gamma ')$.
			
		\end{itemize}
		
	\end{definition}
	
	\noindent \textbf{Computational Correctness}
	A VDSSE scheme $\Psi$ is said to be \emph{correct} if $\forall \lambda \in \mathbb{N}$, $\forall K$ generated using $\mathtt{KeyGen}(1^\lambda)$ and all sequences of  search and update operations on $\gamma$, every search outputs the correct set of identifiers, except with a negligible probability. 
	
	
	\medskip \noindent \textbf{Verifiability}
	Note that, when we are saying a scheme is verifiable, it means that it verifies whether the search result is from the currently updated state of the database according to the owner. Verification does not include update of the database at cloud side. For example, let an owner added a document with some keywords and the cloud does not update the database. Later, if the owner searches with some keywords present in the document and it should get the identifier of the document in the result set. Then, the result can be taken as verified. 
	
	\subsection{Security Definitions}
	We follow security definition of \cite{dcc/SoleimanianK19}. There are two parts in the definition-- confidentiality and soundness.
	We define security in adaptive adversary model where the adversary can send query depending on the previous results.
	Typically, most of the dynamic SSE schemes define its security in this model.
	
	A DSSE, that does not consider verifiability, considers honest-but-curious (HbC) cloud server. In these cases, The owner of the database allows some leakage on every query made. However, it guarantees that no meaningful information about the database are revealed other than the allowed leakages. Soundness definition ensures that the results received form the cloud server are correct.     
	
	\subsubsection{Confidentiality}
	Confidentiality ensures that a scheme does not give any meaningful information other than it is allowed.
	In our model, we have considered the cloud to be malicious. However, the auditor is HbC. Since, verifiability has some monetary cost for the owner, it wants verifiability only when it is required. Also the auditor does not have the database and search ability. Given the proof, it only verifies the result. 
	Thus, if the scheme is secure from cloud, it is so from auditor. Again, we have assumed that the cloud and the auditor do not collude. Hence, we do not consider the auditor in our definition of confidentiality.

	\begin{definition}[CKA2-Confidentiality] \label{def_Confidentiality}
		Let $\Psi =$ $(\mathtt{VKeyGen}$, $\mathtt{VBuild}$, $\mathtt{VSearchToken}$, $\mathtt{ VSearch}$, $\mathtt{VUpdateToken})$ be a verifiable DSSE scheme. Let $\mathcal{A}$, $\mathcal{C}$ and $\mathcal{S}$  be a stateful adversary, a challenger and a stateful simulator respectively. Let $\mathcal{L}$=$(\mathcal{L}_{bld}, \mathcal{L}_{srch}, \mathcal{L}_{updt} )$ be a stateful leakage algorithm.  
		Let us consider the following two games.
		
		\medskip \noindent 
		${\textbf{Real}}_ \mathcal{A}(\lambda)$:
		\begin{enumerate}
			\item The challenger $\mathcal{C}$ generates a key $K \leftarrow \mathtt{VKeyGen}(1^\lambda)$.
			\item $\mathcal{A}$ generates and sends $\mathcal{DB}$ to $\mathcal{C}$. 		
			\item $\mathcal{C}$ builds $(\overline{DB} , \gamma) \leftarrow \mathtt{VBuild}(K,\mathcal{DB})$ and sends $(\overline{DB} , \gamma) $ it to $\mathcal{A}$.			
			\item $\mathcal{A}$ makes a polynomial number of adaptive queries. In each of them, it sends either a search query for a keyword $w$ or an update query for a keyword-document pair $(w,id)$ and operation bit $op$ to $\mathcal{C}$.		
			\item $\mathcal{C}$ returns either a search token $\tau_s \leftarrow \mathtt{VSearchToken}(K,w)$ or an update token  $\tau_u \leftarrow  \mathtt{VUpdateToken}(K,id)$ to $\mathcal{A}$  depending on the query.   		
			\item Finally $\mathcal{A} $ returns a bit $ b $ that is output by the experiment. 
		\end{enumerate}
		
		\medskip \noindent ${\textbf{Ideal}}_ {\mathcal{A},\mathcal{S}}(\lambda)$:
		\begin{enumerate}
			\item $\mathcal{A}$ generates a set $\mathcal{DB}$ of documents and gives it to $\mathcal{S}$ together with $\mathcal{L}_{bld} (\mathcal{DB})$.
			\item $\mathcal{S}$ generates  $(\overline{DB} , \gamma)$  and sends it to $\mathcal{A}$
			\item $\mathcal{A}$ makes a polynomial number of adaptive queries $ q $. For each query,
			$\mathcal{S}$ is given either $ \mathcal{L}_{srch}(w,\mathcal{DB})$ or $\mathcal{L}_{updt}({op,w,id})$ depending on the query.
			\item $\mathcal{S}$ returns, depending on the query $q$, to $\mathcal{A}$ either search token $\tau_s $ or update token  $\tau_u $.  
			\item Finally $\mathcal{A} $ returns a bit $b'$ that is output by the experiment. 
		\end{enumerate}
		
		We say  $\Psi$ is $\mathcal{L}$-secure against adaptive dynamic chosen-keyword attacks if $\forall$ PPT adversary $\mathcal{A}$, $\exists$ a simulator $\mathcal{S}$ such that
		\begin{equation}
		|Pr[{\textbf{Real}}_ \mathcal{A}(\lambda)=1] - Pr[{\textbf{Ideal}}_ {\mathcal{A},\mathcal{S}}(\lambda)=1]| \leq \mu(\lambda)
		\end{equation}
		where $ \mu(\lambda)$ is negligible in $\lambda$.
	\end{definition}
	
	\subsubsection{Soundness}
	
	The soundness property ensures that if a malicious cloud tries to make the owner fool by returning incorrect result it will be caught to the auditor. We define game-based definition of soundness as follows.
	\begin{definition}
		Let  $\Psi$ be a verifiable DSSE scheme with $\Psi =$ $(\mathtt{VKeyGen}$, $\mathtt{VBuild}$, $\mathtt{VSearchToken}$, $\mathtt{ VSearch}$, $\mathtt{VUpdateToken})$. Let us consider the following game.
		
		\medskip \noindent ${\textbf{sound}}_ {\mathcal{A},{\Psi}}(\lambda)$:
		\begin{enumerate}
			\item The challenger $\mathcal{C}$ generates a key $K \leftarrow \mathtt{VKeyGen}(1^\lambda)$.
			\item $\mathcal{A}$ generates and sends $\mathcal{DB}$ to $\mathcal{C}$. 
			\item $\mathcal{C}$ computes $(\overline{DB} , \gamma) \leftarrow \mathtt{VBuild}(K,\mathcal{DB})$ and sends $(\overline{DB} , \gamma)$ to $\mathcal{A}$.		
			\item $\mathcal{A}$ makes a polynomial number of adaptive queries. In each of them, it sends either a search query for a keyword $w$ or an update query for a keyword-document pair $(w,id)$ and operation bit $op$ to $\mathcal{C}$.		
			\item $\mathcal{C}$ returns either a search token $\tau_s \leftarrow \mathtt{VSearchToken}(K,w)$ or an update token  $\tau_u \leftarrow  \mathtt{VUpdateToken}(K,id)$ to $\mathcal{A}$  depending on the query.   		
			\item After making polynomial number of queries, $\mathcal{A}$ chooses a target keyword $w$ and send search query to $\mathcal{C}$.
			\item  $\mathcal{C}$ returns a search token $\tau_s $. $\mathcal{A}$  executes and gets $(R_w, \nu_w )$ where  $\nu_w = accept$ is verification bit from $\mathcal{C}$.
			\item $\mathcal{A}$ generates pair $(R^*_w)$ for a keyword $w$ and gets verification bit $\nu^*_w=accept$.
			\item If $\nu^*_w=accept$ even when $ R^*_w \neq DB(w)$, $\mathcal{A}$ returns $1$ as output of the game, otherwise returns $0$.
		\end{enumerate}
		We say that  $\Psi$ is \emph{sound} if $\forall$ PPT adversaries $\mathcal{A}$, $Pr[{\textbf{sound}}_ {\mathcal{A},{\Psi}}(\lambda) = 1] \leq \mu(\lambda)$.
	\end{definition}
	
	\newpage
	\section{Verifiable SSE with static data} \label{sec_VSSE} 
	Since, in a verifiable SSE scheme, there is no update, it does not have $\mathtt{VUpdate}$ or $\mathtt{VUpdateToken}$ operation. 
	We present a generic scheme that will make any SSE scheme verifiable. 
	Our target is to achieve verifiability, in presence of malicious server, without loosing any other security property with minimal communication and computational costs. 
	
	\subsection{Issues with the existing verifiable SSE schemes}
	There are papers who considered static SSE schemes 
	and suggested authentication tag  generation using MAC to protect the integrity of the search result. 
	For each keyword $w$, they generates a tag  $tag_w  = H(id^w_1|| id^w_2|| \ldots || id^w_{c_w})$ where $H$ is a one-way hash function.   
	Trivially, if the tags are stored at the owner side then the scheme becomes privately verifiable. 
	In that case, when a search is required, the owner can check integrity after receiving the result from the cloud. 
	
	However, this integrity checking does not protect the SSE scheme from malicious adversary \emph{if the tags are outsourced} to the cloud. 
	Checking integrity provides security only from honest-but curious cloud servers. 
	Let us consider an example. 
	Suppose a keyword $w\in \mathcal{W}$ is searched and cloud gets the result $R_w = \{id^w_1, id^w_2, \ldots , id^w_{c_w}, tag_w\} $. 
	Later, if some other keyword $w'$ is searched, the cloud can return the same result and will pass the integrity checking. 
	
	\subsection{A generic verifiable SSE scheme without client storage}
	Since, it is desirable to outsource the data as well as tags to the cloud, the above result shows that checking integrity in the above way can not be considered. 
	It is easy to see that the scheme with checking integrity of the result identifiers are not enough because there is no binding of the keyword with the tags.  
	Here, we present a generic idea that makes any SSE scheme verifiable.
	\paragraph{Scheme Description}
	Let $\Sigma_s = (\mathtt{KeyGen}, \mathtt{Build},\mathtt{SearchToken},\mathtt{Search})$ be a result revealing static SSE scheme. 
	We present a VSSE scheme $\Psi_s$=$(\mathtt{VKeyGen}$, $\mathtt{VBuild}$, $\mathtt{VSearchToken}$, $\mathtt{VSearch})$ for static database as follows.

	Let $H$ be a one-way hash function and a key $K'$ is chosen at random. 
	For each keyword $w\in \mathcal{W}$, a key $k_w= H (K',w)$ is generated.
	$k_w$ is then used to bind the keyword with corresponding tag $tag_w=H({k_w||id^w_1|| id^w_2||\ldots|| id^w_{c_w}})$.
	Finally, for each keyword $w$, $\{id^w_1, id^w_2,\ldots, id^w_{c_w},tag_w \}$ is encrypted at build phase.
	Thus, while performing search with a keyword $w$, as search result, the owner receives $\{id'^w_1, id'^w_2,\ldots, id'^w_{c_w},tag'_w \}$. 
	The owner accepts it if the regenerated tag $tag'_w$ from the received identifiers matched with the received one.
	
	So, the main idea of the scheme is that instead of generating tags only with identifiers, they are bound with $k_w$ which is dependent on $w$ and can be computed by the owner only. After search, if the cloud returns incorrect set of document identifiers then the tag won't get matched.
	The scheme is shown in Fig.~\ref{fig_genericVSSE}. 
	
	Note that, for static case, computing tag is enough to validate a result. Since, one-way hash computation is very efficient and requires small amount of resource, we do not consider any external authority like auditor for verifiability. So, the scheme is privately verifiable.
	\begin{figure} [!htbp]
		\centering	\fbox{\noindent
			\resizebox{1.0\linewidth}{!}{
				\begin{minipage}{0.51\linewidth}	
					\underline{ $\Psi_s.\mathtt{VKeyGen}(1^\lambda)$}
					\begin{enumerate}
						\item $K_{\Sigma_s} \gets  \Sigma_s.\mathtt{KeyGen}(1^\lambda)$
						\item $K'\xleftarrow{\$} \{0,1\}^\lambda $ 
						\item  Return $K_{\Psi_s} = (K', K_{\Sigma_s})$ 
					\end{enumerate}
					
					\underline{ $\Psi_s.\mathtt{VBuild}(\mathcal{DB},K_{\Psi_s})$}
					\begin{enumerate}
						\item $ (K', K_{\Sigma_s}) \gets K_{\Psi_s} $ 			
						\item  \textbf{for} each $w \in \mathcal{W}$
						\begin{enumerate}
							\item $k_w \gets H(K'||w)$
							\item $tag_w \gets H(k_w||id^w_1 || id^w_2 ||\ldots || id^w_{c_w}) $
							\item $DB'(w) \gets DB(w) \cup \{tag_w \} $ 
						\end{enumerate}
						\item $ \mathcal{DB}' \gets \cup  _{w\in \mathcal{W}} DB'(w)  $ 
						\item $(\gamma, \overline{DB}) \gets \Sigma_s.\mathtt{Build} (\mathcal{DB}', K_{\Sigma_s})$  
						\item Return $(\gamma, \overline{DB}) $
					\end{enumerate}
				\end{minipage}
				\begin{minipage}{0.49\linewidth}			
					\underline{$\Psi_s.\mathtt{VSearchToken}(w, K_{\Sigma_s})$}		
					\begin{enumerate}
						\item $\tau_{\Sigma_s} \gets \Sigma_s. \mathtt{ SearchToken}(w, K_{\Sigma_s})  $
						\item Return $ \tau_{\Sigma_s}$
					\end{enumerate}
					
					\underline{ $\Psi_s.\mathtt{VSearch}(\gamma, \tau _{\Sigma_s} )$}
					\begin{enumerate}
						\item  $ (K', K_{\Sigma_s}) \gets K_{\Psi_s} $
						\item $\tau_{\Sigma_s} \gets \Sigma_s. \mathtt{ SearchToken}(w, K_{\Sigma_s})  $
						\item $R_w \gets  \Sigma_s. \mathtt{ Search}(\gamma,\tau_{\Sigma_s}) $  
						\item $k_w \gets H(K'||w) $  
						\item $  \{id'^w_1, id'^w_2,\ldots , id'^w_{c_w}, tag'_w\}  \gets R_w $  
						\item $tag_w \gets H(k_w||id'^w_1 || id'^w_2 ||\ldots || id'^w_{c_w}) $ 
						\item Accept $R_w$ if $tag'_w =tag_w$
					\end{enumerate}
					
				\end{minipage}
			} 
		} 
		\caption{Algorithm for generic verifiable SSE scheme $\Psi_s$ } \label{fig_genericVSSE}	
	\end{figure}
	%
	\paragraph{Cost for verifiability} 
	The cloud storage is increased by $|\mathcal{W}|$ tags. However, depending on the scheme the actual increment might be less than $|\mathcal{W}|$ tags but still it is asymptotically $O(|\mathcal{W}|)$.  
	The communication cost for verification is only increased by one tag from cloud the owner. 
	If we consider computation, to verify a search result, the owner only has to compute a hash value which is very little.   
	\paragraph{Soundness}
	In case the cloud does not want to perform search properly, then it can not get the identifiers and corresponding tag. So, it has to send either random identifiers or identifiers corresponding to other searched keyword. In both case, It cannot be passed verifiability test to the owner. 
	
	\paragraph{Confidentiality}
	The confidentiality of our proposed scheme follows from the security of the embedded SSE scheme.
	
	\section{Our Proposed Forward Secure Publicly Verifiable DSSE scheme} \label{sec_proposedScheme}
	In this section, we propose a simple generic dynamic SSE scheme which is forward secure as well as verifiable. Let  $\Sigma _f$ = $(\mathtt{KeyGen}$, $\mathtt{Build}$, $\mathtt{Search}$, $\mathtt{SearchToken}$, $\mathtt{Update}$, $\mathtt{UpdateToken}  )$ be a result revealing forward secure dynamic SSE scheme.  
	
	It is to be noted that any forward private SSE scheme stores the present state of the database at client side. 
	Corresponding to each keyword, most of them stores the number of documents containing it. Let $ C =\{c_w: w\in \mathcal{W} \}$ be the list of such numbers.
	
	Since, it considers any forward secure scheme $\Sigma_f$, it only adds an additional encrypted data structure to make the scheme verifiable.
	The algorithms of Our proposed scheme are given in Figure~\ref{fig_VDSSE_wcs1}. They are divided into three phases-- initialization, search and update.

	\begin{figure*}[!htbp]
		\frame{\includegraphics[width=\linewidth]{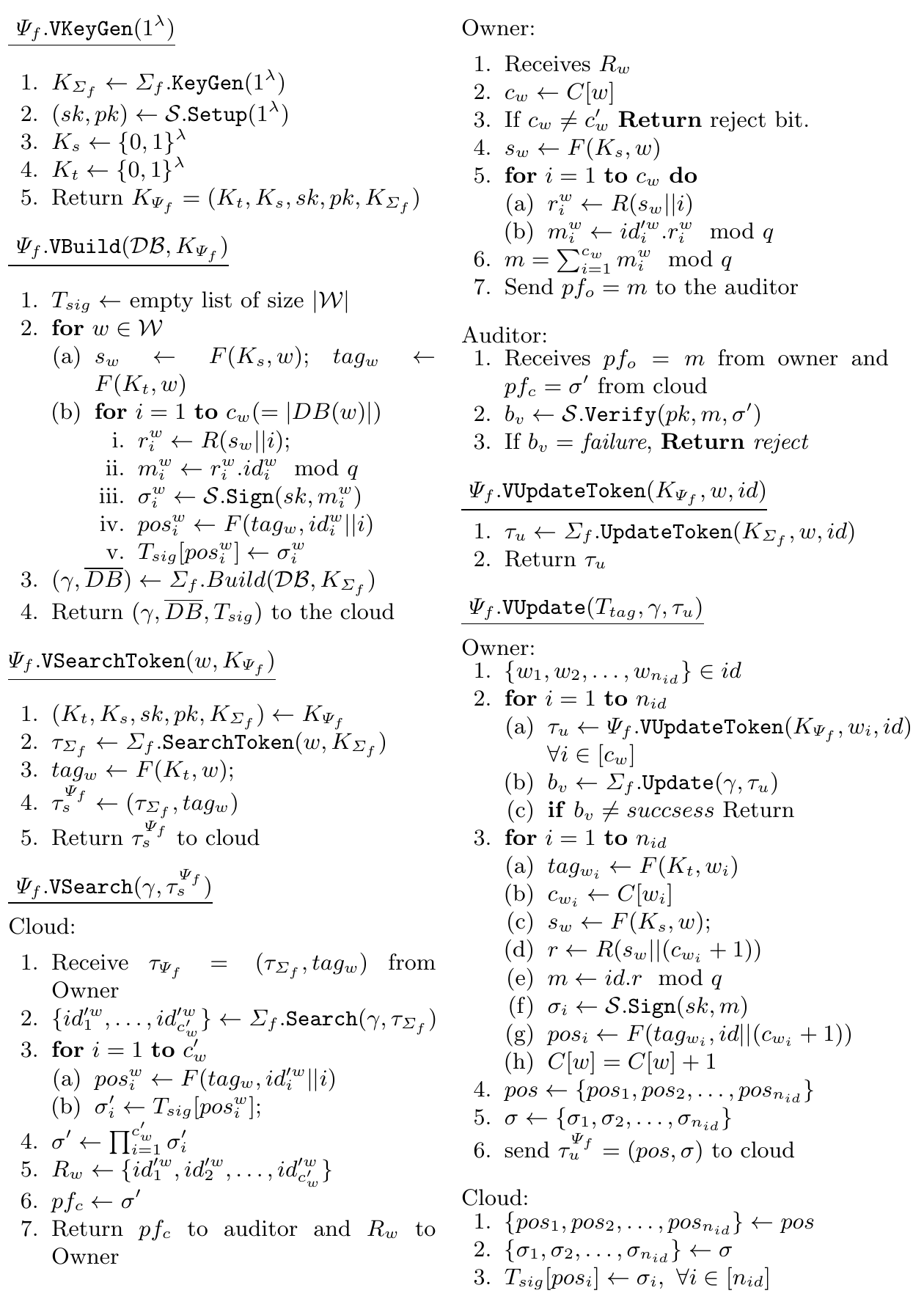}}
		\caption{ Generic verifiable dynamic SSE scheme $\Psi_f$ without extra client storage} \label{fig_VDSSE_wcs1}
	\end{figure*}
	
	\medskip \noindent \textbf{Initialization phase:} 
	In this phase, secret and public keys are generated by the owner and thereafter the encrypted searchable structure is built.  During key generation, three types of keys are generated-- $K_{\Sigma_f}$ for the $\Sigma_f$; $(sk, pk) $ for the bilinear signature scheme; and two random strings $K_s, K_t$ for seed and tag generation respectively. 
	
	Thereafter, a signature table $T_{sig}$ is generated, before building the secure index $\gamma$ and encrypted database $\overline{DB}$, to store the signature corresponding to each keyword-document pair. For each pair $(w, id^w_i) $, the position $pos^w_i = F(tag_w, id^w _i||i)$ is generated with a HMAC $F$. The position is actually act as key of a key-value pair for a dictionary. The document identifier is bounded with $pos^w_i$ together with $tag_w= F(K_t,w)$. The $tag_w$ is fixed for a keyword and is given to the server to find $pos^w_i$. The signature $\sigma^w_i$ for the same pair is also bounded with random number $r^w_i$ which can only be generated from PRG $R$ with the seed $s_w$. Then $(\sigma^w_i, pos^w_i) $ pair is added in the table $T_{sig}$ as key-value pair. After the building process, the owner outsources $\gamma$, $\overline{DB}$ and $T_{sig }$ to the cloud.   
	
	\medskip \noindent \textbf{Search Phase:} 
	In this phase, the owner first generates a search token $\tau_{\Sigma_f} $ to search on $\Sigma_f$. Then, it regenerates $tag_w$ and the seed $s_w$ and then, sends them to the cloud. 
	
	The cloud performs search operation according to $\Sigma_f$ and use the result identifiers $\{id_1 ,id_2 , \ldots id_{c'_w}\}$  to gets the position in $T_{sig}$ corresponding to each pair. It is not able to generate the positions if it does not search for the document identifiers. It collects the signatures stored in those positions, multiplies them and sends multiplication result to the auditor as its part  $pf_c$ of the proof. It sends the search result to the owner. 
	
	The owner first generates random numbers $\{r_1 ,r_2 , \ldots r_{c'_w}\}$  and regenerates aggregate message $m = \sum^{i=c'_w}_{i=1} r_i . id^w_i \mod q$ of the identifiers and sends $m$ to the auditor as $pf_o$, owner's part of the proof. 
	After receiving $pf_c$ and $pf_o$, the auditor only computes $ \mathcal{S}.\mathtt{Verify}(pk, m, \sigma')$. It outputs accept if signature verification returns \emph{success}.
	We can see that the no information about the search results is leaked to the auditor during verification. 
	
	\medskip \noindent \textbf{Update Phase:}
	In our scheme, while adding a document, instead of being updated only a keyword-document pair,  we assume that all such pairs corresponding to the document is added. To add a document with identifier $id$ and keyword set $\{w_1, w_2, \ldots ,w_{n_{id}}\} $, the owner generates the position and the corresponding signature for each containing keyword. The cloud gets them from the owner and adds them in the table $T_{sig}$.

	\medskip \noindent \textbf{Correctness}
	For correctness it is enough to check the following. 
	\begin{equation*}
	\hat{e}(\mathcal{H}(m),pk) = \hat{e}(g^m,g^\alpha)   = \hat{e}(g^{\alpha \sum m_i},g) = \hat{e}(\prod g^{\alpha m_i},g) = \hat{e}(\prod {\sigma_i},g)  = \hat{e}(\sigma,g) 
	\end{equation*}
	
	\medskip \noindent \textbf{Cost for verifiability}
	We achieve, forward privacy as well as public verifiability without client storage in $\Psi_f$. This increases the cloud-storage by $O(N)$, where $N$ is the number of document-keyword pairs. The proof has two parts one from the client and another from the owner. For a keyword $w$, the sizes of them are one group element and one random $\lambda$-bit string only. Thus Auditor receives one element from both. The owner has to compute  $R_w$ integer multiplication and addition, and then has to send one element.   
	
	\medskip \noindent \textbf{Forward privacy}
	We can see that while adding a document, it only adds some keyword-document pair, in the form of key-value pairs. So, During addition, the cloud server is adding key-value pairs in the dictionary. From these pairs, it can not guess the keywords present in it. Again, when it perform searches, it gets about the key (i.e., position on the table) only when it gets the identifiers.
	The one possibility to get the newly added key-value pair linked with the previous is if the added document gives the identifier of it. Since, the one-way function $H'$ gives the document-name of the adding document, the cloud server can not linked it with the previously searched keywords.

	\subsection{Security}
	The security of the scheme is shown in two parts-- confidentiality and soundness.
	
	\medskip \noindent \textbf{Soundness}
	The cloud server can cheat the owner in three ways by sending--
	\begin{enumerate}
		\item Incorrect number of identifiers-- but it is not possible as the owner keeps the number of identifiers.
		\item Same size result of other keywords-- $m$ is generated with a random numbers which can be generated only with the searched keyword and signatures are bound with that. So, the signature verification will be failed.
		\item Result with some altered identifiers-- since signatures are bounded with keywords and the random number, altering any will change $m$ and similarly the signature verification will be failed. 
	\end{enumerate}
	Thus the owner always will get the correct set of document identifiers.

	\subsubsection{Confidentiality}
	Let  $\mathcal{L}^{\Sigma_f} = (\mathcal{L}^{\Sigma_f}_{bld}, \mathcal{L}^{\Sigma_f}_{srch}, \mathcal{L}^{\Sigma_f}_{updt})$ the leakage function of $\Sigma_f$. Let $\mathcal{L}^{\Psi_f} = (\mathcal{L}^{\Psi_f}_{bld}, \mathcal{L}^{\Psi_f}_{srch}, \mathcal{L}^{\Psi_f}_{updt})$ be the leakage function of $\Psi_f$, given as follows.
	\begin{eqnarray*}
		\mathcal{L}^{\Psi_f}_{bld}(\mathcal{DB})  &=& \{ \mathcal{L}^{\Sigma_f}_{bld} (\mathcal{DB}), |T_{sig}|  \} \\
		\mathcal{L}^{\Psi_f}_{srch} (w) &=& \{ \mathcal{L}^{\Sigma_f}_{srch}(w),\{(id^w_i, pos^w_i, \sigma^w_i):i=1,2,\ldots,c_w \} \} \\
		\mathcal{L}^{\Psi_f}_{updt} (f) &=&\{ id,\{ (\mathcal{L}^{\Sigma_f}_{updt}(w_i,id), pos^{w_i}, \sigma^{w_i}):i=1,2,\ldots,n_{id} \} \}
	\end{eqnarray*}
	We show that $\Psi$ is $\mathcal{L}^{\Psi_f}$-secure against adaptive dynamic chosen-keyword attacks in the random oracle model, in the following theorem.
	\begin{theorem}
		If $F$ is a PRF, $R$ is a PRG and $\Sigma_f$ is $\mathcal{L}^{\Sigma_f} $-secure, then  $\Psi_f$ is $\mathcal{L}^{\Psi_f}$-secure against adaptive dynamic chosen-keyword attacks. 
	\end{theorem}
	
	\begin{proof}
		To prove the above theorem, it is sufficient to show that there exists a simulator $\mathtt{Sim}_{\Sigma_f}$ such that $\forall$ PPT adversary $\mathcal{A}$, the output of ${\textbf{Real}}_ \mathcal{A}(\lambda)$ and ${\textbf{Ideal}}_ {\mathcal{A},\mathtt{Sim}_{\Sigma_f}}(\lambda)$ are computationally indistinguishable.
		
		We  construct such a simulator $\mathtt{Sim}_{\Sigma_f}$ which adaptively simulates the extra data structure $T_{sig}$ and query tokens. 
		Let $\mathtt{Sim}_{\Sigma_f}$ be the simulator of the $\Sigma_f$.  
		We simulate the algorithms in Figure~\ref{fig_proof}.
		\begin{figure*}[!htbp]
			\centering
			\frame{\includegraphics[width=1.1\linewidth]{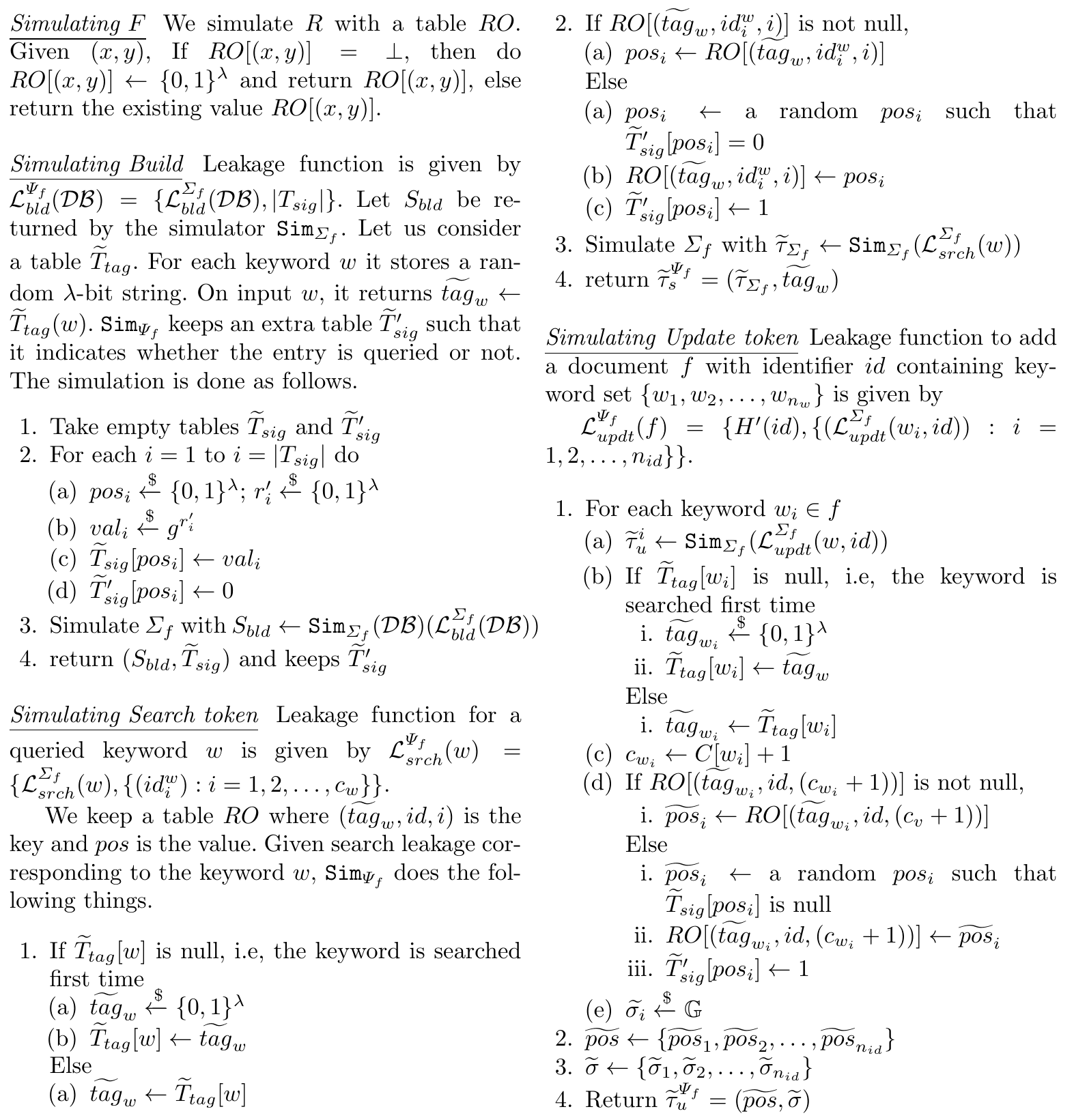}}
			\caption{Simulation of build, search token and update token}\label{fig_proof}
		\end{figure*}
		
		Since, in each entry, the signature generated in $T_{sig}$ is of the form $g^{\alpha mr}$ and  corresponding entry in $\widetilde{T}_{sig}$  is of the form $g^{\alpha r'}$, where $r$ is pseudo-random (as $R$ is so) and $r'$ is randomly taken, we can say that power of $g$ in both are indistinguishable. Hence, $T_{sig}$ and  $\widetilde{T}_{sig}$ are indistinguishable. 
		
		Besides, the indistinguishability of $\widetilde{ \tau}^{\Psi_f}_{u}$, $\widetilde{ \tau}^{\Psi_f}_{s} $ with respect to ${ \tau}^{\Psi_f}_{s}$, ${ \tau}^{\Psi_f}_u $ respectively follows from the pseudo-randomness of $F$. 
		
	\end{proof}

	\subsection{Deletion Support}
	$\Psi_f$ can be extended to deletion support by duplicating it. Together with $\Psi_f$ for addition, a duplicate $\Psi'_f$ can be kept for deleted files. During search, the auditor verifies both separately. The client gets result from both $\Psi_f$ and $\Psi'_f$, accepts only if both are verified and gets the final result calculating the difference.
	\section{Comparison with existing schemes} \label{sec_comparison} 
	Our generic VSSE $\Psi_s$ requires only one hash-value  computation to verify a search which is optimal. Again, during building, the owner  requires $2|\mathcal{W}|$ extra hash-value computation twice of the optimal. We can take that much computation to protect the scheme from malicious server without any extra client storage.  
	
	\begin{table}[!htbp]
		\centering
		\caption{Comparison of verifiable dynamic SSE schemes} \label{tab_comparison}		
		\resizebox{1.0\linewidth}{!}{
			\begin{tabular}{|c|c|c|c|c|c|c|c|c|c|c|}
				\hline
				Scheme 	 	& Forward & Public &\multicolumn{2}{|c|}{Extra Storage} & \multicolumn{3}{|c|}{Extra Computation}  		&  \multicolumn{2}{|c|}{Extra Communication} \\ \cline{4-10}
				Name 	 & privacy & verifiability  & owner 	& cloud & owner 	& cloud & auditor & owner 	& auditor\\ 
				\hline \hline
				Yoneyama and Kimura~\cite{icics/YoneyamaK17}  & $\checkmark$ &$\times$ & $O(|\mathcal{W}|)$ &$O(|\mathcal{W}| log |\mathcal{DB}|) $ & $O(|R_w|)$& $O(|R_w|)$ & -- & $O(1)$ & --\\  	
				\hline 
				Bost and Fouque~\cite{iacr/BostFP16}	 &$\times$  & $\times$  &$O(|\mathcal{W}|)$ & $O(|\mathcal{W}|)$&$O(|R_w|)$ &$O(1)$ &-- & $O(1)$&-- \\ 			
				\hline
				Miao et al.~\cite{isci/MiaoWWM19}	 &$\times$  &  $\checkmark$  &$O(|\mathcal{W}|)$ &$O(N + |\mathcal{W}|)$  & $O(|R_w|)$ & $O(|R_w|)$ & -- & $O(1)$& --\\	
				\hline
				Zhu et al.~\cite{trustcom/ZhuLW16}	 &$\times$  & $\times$  &$O(1)$  &$O(1)$ &$O(|R_w|)$ & $O(|R_w|+ N)$& -- &$O(|R_w|)$ & --\\	
				\hline
				Jiang et al.~\cite{globecom/JiangZGL15} & $\times$ &  $\checkmark$  & $O(1)$ & $O(|\mathcal{W}|)$ &$O(\log |\mathcal{W}|)$ & $O(|R_w|+ N)$& -- &$O(1)$ & --\\	
				\hline 
				$\Psi_f$ & $\checkmark$ & $\checkmark$ & $O(1)$ 			   & $O(N)$  			& $O(|R_w|)$	&$O(|R_w|)$ &$O(1)$&$O(1)$ 	& $O(1)$\\	
				\hline
			\end{tabular}
		}%
		\newline
		Where $N$ is the \#keyword-doc pairs. Here
		extra storage is calculated over all storage,
		extra communication and computation are for a single search.
	\end{table}

	We have compared our verifiable DSSE scheme $\Psi_f$ with verifiable dynamic schemes by 
	Yoneyama and Kimura~\cite{icics/YoneyamaK17}, 
	Bost and Fouque~\cite{iacr/BostFP16},
	Miao et al.~\cite{isci/MiaoWWM19},	
	Zhu et al.~\cite{trustcom/ZhuLW16} and
	Jiang et al.~\cite{globecom/JiangZGL15}. 
	The comparison is shown in Table~\ref{tab_comparison}. From the table, it can be observed that $\Psi_f$  is very efficient with respect to low resource owner. 
	Extra computation needed by the owner, to verify the search, is only $|R_w|$ multiplication which very less from the others.    
	The owner also does not require any extra storage than the built in forward secure DSSE scheme.

	\section{Conclusion} \label{Conclusion}
	Throughout, we have seen that we have successfully presented a privately verifiable SSE scheme and a publicly verifiable DSSE scheme. Both of them are simple and easy to implement. Moreover, the VDSSE scheme achieves forward secrecy.  
	In both of the scheme we have achieved our target to make efficient for low-resource owner. 
	Due to low computational and communication cost, we do need any auditor for VSSE. However, presence of an auditor, who verifies the search result, reduces workload of the owner. 
	Our proposed schemes are only for single keyword search queries. There are many other complex queries too. As a future work, one can design complex queried verifiable DSSE scheme. On the other hand, while designing, keeping them forward secret is also a challenging direction of research.  
	
	
	%
	%
	%


\begin{thebibliography}{10}
		
		\bibitem{cns/AzraouiEOM15}
		Monir Azraoui, Kaoutar Elkhiyaoui, Melek {\"{O}}nen, and Refik Molva.
		\newblock Publicly verifiable conjunctive keyword search in outsourced
		databases.
		\newblock In {\em 2015 {IEEE} Conference on Communications and Network
			Security, {CNS} 2015, Florence, Italy, September 28-30, 2015}, pages
		619--627, 2015.
		
		\bibitem{ccs/Bost16}
		Raphael Bost.
		\newblock {\(\sum\)}o{\(\varphi\)}o{\(\varsigma\)}: Forward secure searchable
		encryption.
		\newblock In {\em Proceedings of the 2016 {ACM} {SIGSAC} Conference on Computer
			and Communications Security, Vienna, Austria, October 24-28, 2016}, pages
		1143--1154, 2016.
		
		\bibitem{iacr/BostFP16}
		Raphael Bost, Pierre{-}Alain Fouque, and David Pointcheval.
		\newblock Verifiable dynamic symmetric searchable encryption: Optimality and
		forward security.
		\newblock {\em {IACR} Cryptology ePrint Archive}, 2016:62, 2016.
		
		\bibitem{ccs/BostMO17}
		Rapha{\"{e}}l Bost, Brice Minaud, and Olga Ohrimenko.
		\newblock Forward and backward private searchable encryption from constrained
		cryptographic primitives.
		\newblock In {\em Proceedings of the 2017 {ACM} {SIGSAC} Conference on Computer
			and Communications Security, {CCS} 2017, Dallas, TX, USA, October 30 -
			November 03, 2017}, pages 1465--1482, 2017.
		
		\bibitem{icc/ChaiG12}
		Qi~Chai and Guang Gong.
		\newblock Verifiable symmetric searchable encryption for
		semi-honest-but-curious cloud servers.
		\newblock In {\em Proceedings of {IEEE} International Conference on
			Communications, {ICC} 2012, Ottawa, ON, Canada, June 10-15, 2012}, pages
		917--922, 2012.
		
		\bibitem{asiacrypt10}
		Melissa Chase and Seny Kamara.
		\newblock Structured encryption and controlled disclosure.
		\newblock In {\em Advances in Cryptology - {ASIACRYPT} 2010 - 16th
			International Conference on the Theory and Application of Cryptology and
			Information Security, Singapore, December 5-9, 2010. Proceedings}, pages
		577--594, 2010.
		
		\bibitem{ccs/ChengYGZR15}
		Rong Cheng, Jingbo Yan, Chaowen Guan, Fangguo Zhang, and Kui Ren.
		\newblock Verifiable searchable symmetric encryption from indistinguishability
		obfuscation.
		\newblock In {\em Proceedings of the 10th {ACM} Symposium on Information,
			Computer and Communications Security, {ASIA} {CCS} '15, Singapore, April
			14-17, 2015}, pages 621--626, 2015.
		
		\bibitem{ccs/CurtmolaGKO06}
		Reza Curtmola, Juan~A. Garay, Seny Kamara, and Rafail Ostrovsky.
		\newblock Searchable symmetric encryption: improved definitions and efficient
		constructions.
		\newblock In {\em Proceedings of the 13th {ACM} Conference on Computer and
			Communications Security, {CCS} 2006, Alexandria, VA, USA, Ioctober 30 -
			November 3, 2006}, pages 79--88, 2006.
		
		\bibitem{globecom/JiangZGL15}
		Shunrong Jiang, Xiaoyan Zhu, Linke Guo, and Jianqing Liu.
		\newblock Publicly verifiable boolean query over outsourced encrypted data.
		\newblock In {\em 2015 {IEEE} Global Communications Conference, {GLOBECOM}
			2015, San Diego, CA, USA, December 6-10, 2015}, pages 1--6, 2015.
		
		\bibitem{ccs/KamaraPR12}
		Seny Kamara, Charalampos Papamanthou, and Tom Roeder.
		\newblock Dynamic searchable symmetric encryption.
		\newblock In {\em the {ACM} Conference on Computer and Communications Security,
			CCS'12, Raleigh, NC, USA, October 16-18, 2012}, pages 965--976, 2012.
		
		\bibitem{ijisec/LiZQLX18}
		Yuxi Li, Fucai Zhou, Yuhai Qin, Muqing Lin, and Zifeng Xu.
		\newblock Integrity-verifiable conjunctive keyword searchable encryption in
		cloud storage.
		\newblock {\em Int. J. Inf. Sec.}, 17(5):549--568, 2018.
		
		\bibitem{fgcs/LiuLLJL18}
		Zheli Liu, Tong Li, Ping Li, Chunfu Jia, and Jin Li.
		\newblock Verifiable searchable encryption with aggregate keys for data sharing
		system.
		\newblock {\em Future Generation Comp. Syst.}, 78:778--788, 2018.
		
		\bibitem{isci/MiaoWWM19}
		Meixia Miao, Jianfeng Wang, Sheng Wen, and Jianfeng Ma.
		\newblock Publicly verifiable database scheme with efficient keyword search.
		\newblock {\em Inf. Sci.}, 475:18--28, 2019.
		
		\bibitem{percom/MiaoMLJZSL17}
		Yinbin Miao, Jianfeng Ma, Ximeng Liu, Qi~Jiang, Junwei Zhang, Limin Shen, and
		Zhiquan Liu.
		\newblock {VCKSM:} verifiable conjunctive keyword search over mobile e-health
		cloud in shared multi-owner settings.
		\newblock {\em Pervasive and Mobile Computing}, 40:205--219, 2017.
		
		\bibitem{chinaf/MiaoMLZL17}
		Yinbin Miao, Jianfeng Ma, Ximeng Liu, Junwei Zhang, and Zhiquan Liu.
		\newblock {VKSE-MO:} verifiable keyword search over encrypted data in
		multi-owner settings.
		\newblock {\em {SCIENCE} {CHINA} Information Sciences},
		60(12):122105:1--122105:15, 2017.
		
		\bibitem{ppna/MiaoMWLWL17}
		Yinbin Miao, Jianfeng Ma, Fushan Wei, Zhiquan Liu, Xu~An Wang, and Cunbo Lu.
		\newblock {VCSE:} verifiable conjunctive keywords search over encrypted data
		without secure-channel.
		\newblock {\em Peer-to-Peer Networking and Applications}, 10(4):995--1007,
		2017.
		
		\bibitem{fc/OgataK17}
		Wakaha Ogata and Kaoru Kurosawa.
		\newblock Efficient no-dictionary verifiable searchable symmetric encryption.
		\newblock In {\em Financial Cryptography and Data Security - 21st International
			Conference, {FC} 2017, Sliema, Malta, April 3-7, 2017, Revised Selected
			Papers}, pages 498--516, 2017.
		
		\bibitem{dcc/SoleimanianK19}
		Azam Soleimanian and Shahram Khazaei.
		\newblock Publicly verifiable searchable symmetric encryption based on
		efficient cryptographic components.
		\newblock {\em Des. Codes Cryptography}, 87(1):123--147, 2019.
		
		\bibitem{ccs/SunYLSSVN18}
		Shifeng Sun, Xingliang Yuan, Joseph~K. Liu, Ron Steinfeld, Amin Sakzad, Viet
		Vo, and Surya Nepal.
		\newblock Practical backward-secure searchable encryption from symmetric
		puncturable encryption.
		\newblock In {\em Proceedings of the 2018 {ACM} {SIGSAC} Conference on Computer
			and Communications Security, {CCS} 2018, Toronto, ON, Canada, October 15-19,
			2018}, pages 763--780, 2018.
		
		\bibitem{infocom/SunLLH015}
		Wenhai Sun, Xuefeng Liu, Wenjing Lou, Y.~Thomas Hou, and Hui Li.
		\newblock Catch you if you lie to me: Efficient verifiable conjunctive keyword
		search over large dynamic encrypted cloud data.
		\newblock In {\em 2015 {IEEE} Conference on Computer Communications, {INFOCOM}
			2015, Kowloon, Hong Kong, April 26 - May 1, 2015}, pages 2110--2118, 2015.
		
		\bibitem{sdm10}
		Peter van Liesdonk, Saeed Sedghi, Jeroen Doumen, Pieter~H. Hartel, and Willem
		Jonker.
		\newblock Computationally efficient searchable symmetric encryption.
		\newblock In {\em Secure Data Management, 7th {VLDB} Workshop, {SDM} 2010,
			Singapore, September 17, 2010. Proceedings}, pages 87--100, 2010.
		
		\bibitem{esorics/Wang0SLAZ18}
		Jianfeng Wang, Xiaofeng Chen, Shifeng Sun, Joseph~K. Liu, Man~Ho Au, and
		Zhi{-}Hui Zhan.
		\newblock Towards efficient verifiable conjunctive keyword search for large
		encrypted database.
		\newblock In {\em Computer Security - 23rd European Symposium on Research in
			Computer Security, {ESORICS} 2018, Barcelona, Spain, September 3-7, 2018,
			Proceedings, Part {II}}, pages 83--100, 2018.
		
		\bibitem{corr/abs-1812-02386}
		Cheng Xu, Ce~Zhang, and Jianliang Xu.
		\newblock vchain: Enabling verifiable boolean range queries over blockchain
		databases.
		\newblock {\em CoRR}, abs/1812.02386, 2018.
		
		\bibitem{icics/YoneyamaK17}
		Kazuki Yoneyama and Shogo Kimura.
		\newblock Verifiable and forward secure dynamic searchable symmetric encryption
		with storage efficiency.
		\newblock In {\em Information and Communications Security - 19th International
			Conference, {ICICS} 2017, Beijing, China, December 6-8, 2017, Proceedings},
		pages 489--501, 2017.
		
		\bibitem{icws/0016XYL16}
		Rui Zhang, Rui Xue, Ting Yu, and Ling Liu.
		\newblock {PVSAE:} {A} public verifiable searchable encryption service
		framework for outsourced encrypted data.
		\newblock In {\em {IEEE} International Conference on Web Services, {ICWS} 2016,
			San Francisco, CA, USA, June 27 - July 2, 2016}, pages 428--435, 2016.
		
		\bibitem{usenix/ZhangKP16}
		Yupeng Zhang, Jonathan Katz, and Charalampos Papamanthou.
		\newblock All your queries are belong to us: The power of file-injection
		attacks on searchable encryption.
		\newblock In {\em 25th {USENIX} Security Symposium, {USENIX} Security 16,
			Austin, TX, USA, August 10-12, 2016.}, pages 707--720, 2016.
		
		\bibitem{trustcom/ZhuLW16}
		Xiaoyu Zhu, Qin Liu, and Guojun Wang.
		\newblock A novel verifiable and dynamic fuzzy keyword search scheme over
		encrypted data in cloud computing.
		\newblock In {\em 2016 {IEEE} Trustcom/BigDataSE/ISPA, Tianjin, China, August
			23-26, 2016}, pages 845--851, 2016.
		
	\end{thebibliography}
\end{document}